\documentclass[letterpaper,12pt]{amsart}
\usepackage[usenames]{color}
\usepackage{xifthen}
\newboolean{colour}
\setboolean{colour}{false}

\usepackage{graphicx}
\usepackage{xspace}
\usepackage{url}
\usepackage{setspace}
\usepackage{natbib}
\usepackage[normalem]{ulem}
\usepackage{dsfont}

\title[Profile mixture model identifiability]{Parameter identifiability for a profile mixture model of protein evolution}

\author{Samaneh Yourdkhani}
\address{Department of Mathematics and Statistics\\
	University of Alaska Fairbanks, 99775}
\email{syourdkhani@alaska.edu}

\author{Elizabeth S. Allman}
\address{Department of Mathematics and Statistics\\
	University of Alaska Fairbanks, 99775}
\email{e.allman@alaska.edu}

\author{John A. Rhodes}
\address{Department of Mathematics and Statistics\\
	University of Alaska Fairbanks, 99775}
\email{j.rhodes@alaska.edu}

\date{June 30, 2020}

\usepackage{latexsym,array,delarray,amsthm,amssymb,epsfig,verbatim,tikz}
\usetikzlibrary{arrows, positioning, automata}
\usepackage[latin1]{inputenc}
\usepackage{tikz}
\usetikzlibrary{trees}
\usepackage{verbatim}
\usetikzlibrary{arrows, positioning, automata}
\usepackage{caption,tikz,adjustbox,pdfpages}
\usepackage{subcaption}
\usetikzlibrary{decorations.pathreplacing}

\usetikzlibrary{calc, through, shapes, positioning, decorations,decorations.markings}
\usetikzlibrary{arrows.meta,decorations.markings,calc,positioning}
\tikzset{leaf/.style={fill=white,inner sep=1pt}}
\tikzset{label/.style={midway,fill=white,inner sep=2pt}}
\tikzset{vertex/.style={draw,circle,fill=black,inner sep = 1pt, minimum size=4pt}}
\usetikzlibrary{shapes.geometric}
\usetikzlibrary{positioning}

\theoremstyle{plain}
\newtheorem{theorem}{Theorem}[section]
\newtheorem{lemma}[theorem]{Lemma}
\newtheorem{prop}[theorem]{Proposition}
\newtheorem{definition}[theorem]{Definition}
\newtheorem{cor}[theorem]{Corollary}
\theoremstyle{definition}

\theoremstyle{remark}

\newtheorem{example*}{Example}

\newcommand{\CC}{\mathbb{C}}

\newcommand{\rank}{\mathrm{rank}}
\newcommand{\diag}{\operatorname{diag}}
\newcommand{\pari}{\texttt{Pari/GP}\xspace}

\usepackage{xifthen}
\newboolean{submittedVersion}
\setboolean{submittedVersion}{false}

\begin{document}

\ifthenelse{\boolean{submittedVersion}}{\doublespacing}{}

\begin{abstract} 
  A Profile Mixture Model is a model of protein evolution,
  describing sequence data in which sites are assumed to follow
  many related substitution processes on a single evolutionary
  tree.   The processes depend in part on different amino acid 
  distributions, or profiles, varying over sites in aligned sequences.
  A fundamental question for any stochastic model, which must be
  answered positively to justify model-based inference, is whether the
  parameters are identifiable from the probability distribution they
  determine. Here we show that a Profile Mixture Model has
  identifiable parameters under circumstances in which it is likely to be
  used for empirical analyses. In particular, for a tree relating 9 or more
  taxa, both the tree topology and all numerical parameters are
  generically identifiable when the number of profiles is less than
  74.
\end{abstract}

\maketitle


\ifthenelse{\boolean{submittedVersion}}{
\vskip 1in

Contact Information:

\

Samaneh Yourdkhani \  (corresponding author)

Department of Mathematics and Statistics

University of Alaska Fairbanks, 99775

syourdkhani@alaska.edu

\bigskip

Elizabeth S. Allman

Department of Mathematics and Statistics

University of Alaska Fairbanks, 99775

e.allman@alaska.edu

907-474-2479

\bigskip

John A. Rhodes

Department of Mathematics and Statistics

University of Alaska Fairbanks, 99775

j.rhodes@alaska.edu

907-474-5445

\newpage
}{}


\section{Introduction} 
A Profile Mixture model is a certain stochastic model of protein sequence evolution that describes the 
changes in sequences along the tree of evolutionary relationships of a collection of taxa. 
Such a model is often used for the inference of the tree from sequence data, using standard maximum 
likelihood or Bayesian statistical frameworks. Here we investigate the question of  \emph{parameter identifiability} 
for this model: Are the model parameters --- both the tree topology and numerical ones --- determined by a site 
pattern distribution arising from the model? Parameter identifiability, 
which informally means that valid parameter inference is possible in
ideal circumstances,
is an essential component of the theoretical 
justification for standard statistical inference approaches.

\medskip

In models of protein sequence generation, amino acid site patterns 
are generally  assumed to be independent and identically
distributed across the sites.  Common continuous-time models of amino acid substitutions are instances of  
the \emph{general time-reversible model} (GTR) which assumes a single rate matrix  $Q$ constant over a 
metric tree, or extensions that allow for additional scalar rate variation at individual
sites.  The rate matrix
$Q$ has off-diagonal entries from $R \, \diag ( \boldsymbol \pi )$, where $R$ is a symmetric matrix of \emph{exchangeabilities} 
and $\boldsymbol \pi$ is a vector of frequencies of the amino acids which remains stable under the model. 

In principle, one can infer $R$, $\boldsymbol \pi$, and a metric tree of taxon relationships from protein 
sequence data using standard statistical frameworks.
However, with 20 amino acids the state space for the model is large, so an exchangeability matrix $R$ is often fixed
in advance, having been previously determined empirically for particular types of data.  Well-known 
exchangeabilities for protein alignments include the JTT
\citep{JTT92},  WAG  \citep{WG01}, and  LG 
\citep{QGL2008} matrices.

When inspecting protein sequence data, however, it is often clear that the GTR assumption of
identically distributed sites is a poor one, since sites have visibly
different amino acid compositions.   Site residue distributions, or \emph{profiles}, 
likely differ because of biophysical properties of amino acids (e.g., hydrophilia,
polarity, or charge), and the associated structural and functional
constraints on the protein.
This phenomenon suggests a model with multiple classes of substitution
processes, and in particular  a mixture model using a variety of profiles with the same exchangeabilities
for all classes.   Mixture models can provide better fit to data as they introduce more parameters, 
though they also increase computational time and may lead to overfitting of the data.

But a more fundamental issue with adopting a mixture model is that one may lose parameter identifiability.
If several choices, or even more worrisome, infinitely many choices of parameters lead to the same
probability distribution under the model, then even with an idealized infinite data set perfectly in accord 
with the model one could not recover the parameter values under which the data arose. Since the goal of most 
phylogenetic analyses is to infer model parameters --- generally the topological tree but often numerical 
parameters as well ---
identifiability is an essential property for a model to be useful. Non-identifiability poses particular 
challenges in Bayesian MCMC analyses, where it may be manifested as a lack of convergence \citep{RanalaIdent}.

For non-mixture site substitution models in phylogenetics parameter identifiability has long been established, but
mixture models provide greater challenges. Although computational work may suggest whether it holds or fails, 
parameter identifiability can only be established theoretically as it is a model property, and not dependent on an 
inference method. In recent years algebraic methods have been introduced and successfully applied to  a number of 
phylogenetic mixture models,
see, for example, \citet{AR2006, ARGMI, Allman2009a, AHRfilter, APRS, ALR2019,Chifman2015,
LS2015, Hollering2019, KWsvdQ_2020}. While one of these works \citep{RS2012} established a  rather general result on 
parameter identifiability of phylogenetic mixture models with many components, it unfortunately does not apply to the 
profile mixture model's specific structure. 

\medskip

In this work, we prove parameter identifiability for a \emph{Profile
Mixture Model (PM)} of amino acid site substitution.  PM models were introduced in the Bayesian context 
\citep{LP04,PhyloBayes, PhyloBayesMPI} where the number of profiles might be inferred 
using a Dirichlet process prior, and as finite mixtures with a fixed number of components
in a Maximum Likelihood analysis \citep{QGL2008}.  Studies suggest
that PM models perform better than single-class models,
particularly on data that is saturated or with an underlying  long branch attraction 
bias \citep{LartillotEtAl2007, WangLiSuskoRoger2008}.
Mixtures with as many as 60 classes 
have been investigated with empirical data sets, with indications that around 20 profiles often provides good fit \citep{QGL2008}.  
For a recent study assessing the performance under
simulation of mixture models including discrete-$\Gamma$ rates-across-sites and 
PM models, see \citet{WangSuskoRoger2014}.

\smallskip
 
Our main result, Theorem \ref{thm:mainThm}, establishes 
that parameters of a profile mixture model with up to 73 classes on a tree of 9 or more taxa, 
are \emph{generically identifiable}; that is, identifiable outside an exceptional parameter
set of measure zero.  For any fixed number of classes, the parameters include the tree topology, 
the tree's edge lengths, the exchangeabilities, the
profiles, and weights of the mixture components.

The proof techniques we employ are algebraic in nature, using ideas from 
tensor decomposition and algebraic geometry.  These tools, which have been introduced and used previously for phylogenetic models \citep{AR2006, Allman2009a,RS2012}, are based in the algebraic properties of matrices and 3-way tensors obtained from rearranging the entries of the
distribution of site pattern frequencies. However, the structure of the PM model, with profiles varying over classes while the exchangeabilities do not,
introduce important differences that prevent any easy deduction of the result from previous work. At several points in our arguments we use exact integer computation,
performed by the software \pari \citep{PARI2}, to establish certain generic conditions we need on ranks of matrices.

As motivated by applications to amino acid models, our main theorem is stated for the profile mixture model with a state
space of size 20. However, the techniques used for establishing 
it apply to arbitrary sizes $\kappa$ of the state space.
For example, $\kappa$ might be $4$ for DNA, or $61$ for codons.  However,
appropriate rank computations would need to be carried out to complete the proof in such contexts.
In the $\kappa=20$ setting we also believe the proof techniques could be pushed to establish identifiability for more than 73 profiles, at the expense of requiring more taxa on the tree.

\smallskip

This paper is organized as follows: In Section \ref{sec:ModelsOnTrees} 
we introduce phylogenetic substitution models, and in particular the profile
mixture model under study. Section \ref{sec:AlgDefs} provides algebraic definitions
and lemmas, though removed from the biological setting of interest. 
Section \ref{sec:AlgPM} then connects the phylogenetic profile mixture model with
these algebraic notions. We conclude in Section \ref{sec:MainThm} with the proof of
our main theorem on identifiability of the PM model parameters.

\section{Markov Models on Trees} \label{sec:ModelsOnTrees}

We begin by introducing Markov models of site substitution along a tree.
Throughout, let $\kappa$ be the size of the state space, which we identify with
$[\kappa ]=\{1,2,3, \dots, \kappa\}$.
For protein data, $\kappa = 20$.  Let $T^\rho$ be a rooted 
topological tree, with root $\rho$ and leaves labelled by elements of the taxon set $X$.  
The \emph{general Markov model} of $\kappa$-state sequence evolution 
along $T^\rho$ is parameterized by 1) A $1 \times \kappa$
vector $\boldsymbol{\pi}$ giving the distribution
of states at the root; and 2) for each edge $e$ directed away from the
root, a $\kappa \times \kappa$ Markov matrix $M^e$ giving the conditional probabilities of state transitions along $e$.
These determine the expected site pattern frequency array, or joint distribution of states at the leaves,
which we view as a $\underbrace{\kappa \times \kappa \times\dots \times \kappa}_n$ array or tensor,
$P$.
Each site in an alignment is modeled as independent and identically 
distributed according to $P$.

A subclass of general Markov models is composed of the \emph{general time-reversible models (GTR)}.
For a GTR model, there is an single underlying rate matrix $Q$, and for each edge $e$ of $T^\rho$ a length $t_e$ with $M^e = \exp(Qt_e)$.  Time-reversibility is the assumption that for some symmetric $\kappa \times \kappa$ matrix $R$ of 
non-negative \emph{exchangeabilities} and the root distribution $\boldsymbol \pi$  the 
diagonal entries of $Q$ are those of the product $R \diag(\boldsymbol \pi)$, with the 
diagonal entries chosen so that row sums are zero.  This results in
$\diag(\boldsymbol{\pi})Q=Q^T\diag(\boldsymbol{\pi})$.
One consequence of time-reversibility is that the Markov matrix $M^e$ is independent 
of the direction of $e$.  It follows that the tree parameter in a GTR model is 
\emph{de facto} unrooted since the location of the root is not identifiable.  We repeatedly take advantage of this to `move the
root' to locations in $T$ convenient for our arguments.

Profile mixture models are finite mixtures of GTR models, where the underlying exchangeability matrix $R$
is the same for each class.  The particular profile mixture model examined here has parameters
as follows.  

\begin{definition}\label{def:PM}
	Let $T$ be a rooted topological tree, $\kappa\ge 2$ a number of states, and $m\ge 1$ a number of classes. 
	Then the numerical parameters of the \emph{Profile Mixture Model} on $T$, PM=PM $(T, \kappa, m)$, are:
	\begin{enumerate}
		\item a collection of non-negative branch lengths $\{t_e\}$, one for each edge $e$ of $T$;
		\item a symmetric $\kappa\times\kappa$ matrix $R$ of non-negative exchangeabilities;
		\item a collection of $m$ class weights $\{w_i\}$, with $w_i>0$ and $\sum w_i=1$; and
		\item For each class $i = 1,2, \dots, m$,
		\begin{enumerate}
			\item[$-$] a $1 \times \kappa$ root distribution vector $\boldsymbol{\pi}_i$, called a \emph{profile}; and
			\item[$-$] a scalar rate parameter $r_i\geq 0$.		
		\end{enumerate}

	\end{enumerate}  
\end{definition}

The scalar rate parameters $\{r_i \}$ are used to incorporate across-site rate variation into the PM model.
Specifically, for class $i$ with $Q_i$ the rate matrix  determined by $R$, $\boldsymbol{\pi}_i$, the Markov matrix on edge $e$ in $T$ is $M_i^e=\exp(r_iQ_it_e)$.  
We note that site rate variation for PM models
may be implemented differently in software, with a rate for each site
\citep{LP04} or with a discrete-$\Gamma(4)$ \citep{QGL2008}.  In the first implementation, the PM
model is very likely overparameterized and ideally the MCMC would limit the number
of rate multipliers.  Implementation of the rate variation using a discrete-$\Gamma$
has a long history in computation phylogenetics \citep{YangRAS1994}, but proofs
of such rate variation identifiability are only known for the continuous $\Gamma$
\citep{AllmanAneRhodes07, Chai2011}. 

While probability distributions from mixture models are often described as weighted sums of distributions from the various classes,
phylogenetic mixture models can be equivalently presented as a single model on a tree $T$
with $m \kappa$ states at internal nodes of $T$, and $\kappa$ states at the leaves.  The internal states are pairs $(i,j)$ where $i$ is a class and $j\in [\kappa]$ is 
a `usual' state.
In this
formulation, Markov matrices on internal edges $e$ for the PM model are $m \kappa \times
m \kappa$ block diagonal matrices, where the the $m$ blocks are the $M_i^e$, $i = 1, \dots, m$. The block structure prevents changes from one class to another, though the `usual' states may change within the class.
For the terminal edges $e$ of $T$, leading to leaves where the class information is not observable,
the PM Markov matrix for an edge is formed by stacking the $m$ Markov matrices $M^e_i$ for the
classes. The root distribution is an $m\kappa$ vector formed by concatenating  $w_i\boldsymbol \pi_i$ for the classes.

We collect these observations for parameterizing the PM model on a tree.

\begin{definition}\label{def:PM_lambda_kappa}  Given 
	parameters for the profile mixture model $PM(T, \kappa, m)$, assume
	that $T$ is rooted at $r$.   Then the $1 \times m \kappa$ vector $\boldsymbol \Pi= \boldsymbol
\Pi_r = (w_1\boldsymbol{\pi}_1, \, w_2\boldsymbol{\pi}_2, \, \dots, \, w_{\kappa}\boldsymbol{\pi}_{\kappa})$,
the $m \kappa \times m \kappa$ matrices $M^e = exp( Q t_e )$ where $Q$ is block diagonal with blocks
$r_i Q_i$ for each internal edge $e$ of length $t_e$, and the $m \kappa \times \kappa$ matrix $M^e$ formed by
stacking the matrices $M_i^e$ for each class $i$ on a terminal edge
give a parameterization of the PM model as a Markov model of site substitution on $T$.
\end{definition}

Since our main goal is to prove parameter identifiability for the PM model, we formally define
the notion of generic identifiability.

\begin{definition}\label{def:ID}
	Consider a parametric model, specified by a parameterization map $\phi$ from some parameter space to a space of probability distributions. If $\phi$ is one-to-one, then the model parameters are \emph{identifiable}. If $\phi$ is 
	one-to-one except possibly on a subset of measure zero in the parameter space, then the model parameters are \emph{generically identifiable}. 
\end{definition}

It is well known that for the GTR model some normalization is needed for rates and branch lengths since 
$Q t = (s Q) \left(\frac{t}{s} \right)$ shows rescaling all rates in $Q$ can be offset by decreasing branch lengths.  
Once understood and addressed, this model 
overparameterization, or lack of identifiability, is of little consequence.
Typically, the rate matrix $Q$ is normalized so that branch lengths are measured 
in expected number of substitutions per site over the elapsed time.  In the strictest sense, only the normalized
variant of the GTR model has identifiable parameters, a result used in our proof of the main theorem.
\begin{theorem} \label{thm:GTRid} 
	For a single class GTR model on an unrooted metric tree, the tree topology and all numerical 
	parameters are generically identifiable, up to a normalization of $Q$.
\end{theorem}

\section{Algebraic Definitions and Lemmas} \label{sec:AlgDefs}

In this section we collect  algebraic definitions and theorems that will play a role in our 
analysis of the PM model. We present these in a purely algebraic setting, deferring the 
connection to the phylogenetic models, and in particular the PM model, to later sections.
We begin by defining tensors and certain algebraic operations on them leading up to a theorem of J. Kruskal
on the structure of 3-way tensors, an important tool that we will use several times. 
We then briefly introduce algebraic varieties and conclude by stating a theorem for 
identifying generic properties, a tool also used repeatedly in our proofs.

\subsection{Tensors}

Our first definition is a standard one.

\begin{definition}\label{Def3.6}
	Let $A$ be an $m\times k$ matrix and $B$ be an $n\times l$ matrix. 
	The \emph{tensor}, or \emph{Kronecker, product} $A\otimes B$ is the $mn\times kl$ 
	matrix whose rows are indexed by the ordered pair $(i_1,j_1)$, $i_1\in[m], j_1\in[n]$ 
	and whose columns are indexed by ordered pair $(i_2,j_2)$, $i_2\in[k], j_2\in[l]$   
	such that the $((i_1,j_1),(i_2,j_2))$ entry is
	$\left(A\otimes B\right)_{(i_1,j_1),(i_2,j_2)}=a_{i_1 i_2}b_{j_1 j_2}.
	$
\end{definition}

Less standard is the following.

\begin{definition} \label{def:row-tensor-products} \label{def:rowTensorProduct}
	Let $A$ be an $m\times c_1$ matrix and $B$ be an $m\times c_2$ matrix. The 
	\emph{row tensor product} $A \otimes_r B$ is the $m\times c_1 c_2$ matrix
	with entries indexed by 
	$(i,(j,k))$ for $i \in [m], j \in [c_1], k \in [c_2]$,
	$$
	\left(A\otimes_r B\right)_{i,(j,k)}=a_{ij}b_{ik}.
	$$ 
	In the case that $A = B$ and $\ell$ is a positive integer, then the 
	$\ell^{\text{th}}$ \emph{row-tensor power} of $A$ is the $m\times k^{\ell}$ matrix
	$
	A^{\otimes_r^{\ell}}=\underbrace{A\otimes_r A\otimes_r \dots \otimes_r A}_{\ell}.
	$ 
\end{definition}
We do not specify the precise order of row and column indices in these tensor products, 
since for our applications it will either be clear from context, or inconsequential.
In particular, we often only need results on the ranks of these products, which are independent of row and column ordering.

\smallskip

Since Kruskal's Theorem concerns 3-way tensors, we next describe reformatting
$n$-way tensors into $3$-way ones.  Suppose $P$ is an $n$-way tensor
with indices labeled by $X$.  Then a \emph{tripartition} $I | J | K$ of $X$ 
is a collection three disjoint non-empty subsets of $X$ whose union is $X$, $X = I \sqcup J \sqcup K$.
A \emph{bipartition of $X$}, or a \emph{split}, is defined similarly, with the disjoint sets required to be non-empty.

\begin{definition} \label{def:flattening}  
	Let $A$ be an $n$-way $\kappa \times \dots \times \kappa$ tensor with $I|J$ a split of the index set $X$. 
Then the \emph{matrix flattening} of $A$ with respect to $I,J$, denoted $\operatorname{Flat}_{I|J}(A)$, 
is a $\kappa^{|I|}\times \kappa^{|J|}$ matrix. If, by permuting indices, we assume that 
$I=\{1,2,\cdots, |I|\}$, $J=\{|I|+1,\cdots, n\}$, 
then the $(\bf i,\bf j)$-entry is
	$$
	\left(\operatorname{Flat}_{I|J}(A)\right)_{\bf i,\bf j}=
	A(i_1, \dots, i_{\vert I \vert}, \, j_1, \dots, j_{\vert J \vert}),
	$$
	for $\mathbf i = (i_1,\dots,i_{|I|})$ and $\mathbf j = (j_1,\dots,j_{|J|} )$.

Similarly for a tripartition $I|J|K$ of $X$, the $3$-way tensor $\operatorname{Flat}_{I|J|K}(A)$ is
$$
\left(\operatorname{Flat}_{I|J|K}(A)\right)_{\mathbf i,\mathbf j,\mathbf k}=A(\mathbf i,\mathbf j,\mathbf k),
$$ 
where $\mathbf i\in [\kappa]^{|I|}$, $\mathbf j \in [\kappa]^{|J|}$, and $\mathbf k\in [\kappa]^{|K|}$.
\end{definition}

\begin{example*}
	Suppose $A$ is a $20\times20\times20\times20\times20\times20$ $6$-way tensor, and 
	let $I=\{1,3\}$, $J=\{4\}$, and $K=\{2,5,6\}$. Then $\operatorname{Flat}_{I|J|K}(A)$ 
	is a $20^2\times20\times20^3$ tensor with, for example,
	$$
	\left(\operatorname{Flat}_{I|J|K}(A)\right)_{(10,12),(8),(15,16,18)}=A(10,15,12,8,16,18).
	$$
\end{example*}

Kruskal's theorem requires the notion of a $3$-way tensor obtained as sum of  ``outer products" 
of the rows of $3$ matrices.

\begin{definition} \label{def:3wayProd}  
	Let $A$ be a $k \times n_A$ matrix with $i^{th}$ row $r_{i}^A=\left(r_{i}^A(1),\cdots,r_{i}^A( n_A)\right)$, and
	similarly for matrices $B$ and $C$  of size $k\times n_B$ and $k\times n_C$ respectively. Then $[A,B,C]$ 
	denotes the $3$-way  $n_A\times n_B\times n_C$ tensor
	$$[A,B,C]=\sum_{i=1}^{k} r_{i}^A\otimes r_{i}^B \otimes r_{i}^C,
	$$
	where the tensor products in the summands are formatted to preserve an index
	for each matrix.  For instance, $r_1^A \otimes r_1^B = (r_1^A)^T \cdot r_1^B$
	is $n_A \times n_B$, where $T$ denotes the transpose.
\end{definition}

To illustrate, suppose that $A, B, C$ are $2\times 2$, $2\times 3$, and $2\times 4$ matrices respectively,
	$$
	A=\begin{pmatrix}
	1&2\\
	3&4
	\end{pmatrix}, \ \ \ B=\begin{pmatrix}
	1&2&3\\
	4&5&6
	\end{pmatrix}, \ \ \ C=\begin{pmatrix}
	1&2&3&4\\
	5&6&7&8
	\end{pmatrix}.
	$$
Then $P=[A,B,C]$ is the $2\times3\times4$ tensor with slices with respect
to the $C$ index given by
	\begin{align*}
	P(\cdot,\cdot,1)&=\begin{pmatrix}
	61&77&93\\
	82&104&126\\
	\end{pmatrix}, \ \ \ P(\cdot,\cdot,2)=\begin{pmatrix}
	74&94&114\\
	100&128&156\\
	\end{pmatrix},\\
	P(\cdot,\cdot,3)&=\begin{pmatrix}
	87&111&135\\
	118&152&186\\
	\end{pmatrix}, \ \ \ P(\cdot,\cdot,4)=\begin{pmatrix}
	100&128&156\\
	136&176&216\\
	\end{pmatrix}.
\end{align*}

As a simple extension of Definition \ref{def:3wayProd} for use with
phylogenetic models, we write $$[\boldsymbol{\pi};\, A,B,C] 
= [\diag(\boldsymbol \pi) A, B, C ] = 
\sum_{i=1}^{k} \pi_i \, r_{i}^A \otimes r_{i}^B \otimes r_{i}^C,$$
where $\boldsymbol{\pi}=(\pi_1,\pi_2,\cdots,\pi_k).$

Before the stating Kruskal's Theorem, we need the following. 
\begin{definition}   \label{def:KruskalRowRank}   
	Let $A$ be a matrix. The \emph{Kruskal (row) rank} of a matrix
	$A$ is the largest number $k$ such that every set of $k$ rows of $A$ are independent.
\end{definition}

For example, letting $V$ denote
the set of all ${3\times3}$ matrices, a set of dimension $9$, 
consider matrices 	of the form

\begin{equation}\label{eq:Krank}
\begin{pmatrix}
	a&b&c\\ 
	a&b&c\\
	d&e&f
\end{pmatrix},
\end{equation}
	where $(a,b,c), (d,e,f)$ are  
	independent.
	These matrices have rank 2 but  Kruskal rank $1$,  and form a subset of lower dimension 
	inside the $9$-dimensional space $V$.

It is clear that Kruskal rank is less than or equal to matrix rank, but 
when a matrix has full row rank, the two notions coincide.  In subsequent sections,
we exploit this observation by creating matrices with full row rank and therefore full Kruskal rank.

\smallskip

Kruskal's theorem can be viewed as a generic identifiability theorem for 3-way arrays,
showing that triple products satisfying a particular rank condition are decomposable
in essentially a unique way.

\begin{theorem}[\cite{Kruskal77}] \label{thm:Kruskal}  
	Let $A, B, C$ be $l \times n_A$, $l \times n_B$, and $l\times n_C$ matrices 
	with Kruskal rank $p, q, r$ respectively. If
	\begin{equation} \label{eq:Kruskal} 
	p+q+r \, \geq \, 2l+2,
	\end{equation}
	then $A, B, C$ are uniquely determined by $[A, B, C]$, up to simultaneous permutation
	and scaling of their rows. More precisely, if $[A, B, C]=[A', B', C']$ then there exist 
	invertible diagonal matrices $D_1, D_2$ and a permutation matrix $P$ such that 
	$$A'=PD_1A, \ \ \  B'=PD_2B, \ \ \ C'=PD_1^{-1}D_2^{-1}C.
	$$
\end{theorem}
By way of contrast, note that for two compatible matrices $A$, $B$, the natural analog of the 
bracket product is the matrix product $[A,B]=A^T B$. However, from $[A,B]$, $A$ and $B$ 
can \emph{not} be determined uniquely, since there are many matrix products that give the same result.
For instance, $A^T B=(QA)^T (QB)$ for any orthogonal matrix $Q$. Kruskal's theorem thus states 
a significant difference between matrices and $3$-way tensors.

\subsection{Generic points in parameter space}
Algebraic geometry provides a convenient tool for understanding exceptional sets,
like those that fail to satisfy the rank conditions necessary to apply Kruskal's
Theorem.  We briefly give the needed definitions.

\begin{definition} \label{def:variety} 
	Let $S$ be a finite set of polynomials in $\mathbb{C}[x_1,\dots,x_n]$. 
	The common zero set in $\mathbb{C}^n$ of the polynomials in $S$ is 
	the \emph{algebraic variety} $V(S)$. 
        A subset of a variety that is itself a variety is called a \emph{subvariety}.  
       For any algebraic variety $V(S) \subseteq \CC^n$, the \emph{ideal $I(V(S))$} 
	is the set of all 
	polynomials $f\in \mathbb{C}[x_1, \dots, x_n]$ such that $f(v)=0$ for all $v\in V(S)$. 

\end{definition}

The main result of this work is that PM model parameters are identifiable except for 
`rare' choices. This is expressed using the following terminology. 
\begin{definition} \label{def:generic}  
	A property is \emph{generic} on a full-dimensional subset $W$ of $\mathbb R^n$ or $\mathbb C^n$ 
	if it holds at all points of $W$ except possibly for those points in some subset $U \subset W$
	of measure 0. If $V$ is an algebraic variety in $\mathbb C^n$, we say a property is \emph{generic} 
	on $V$ if it holds at all points except those in a proper subvariety of $V$.    
\end{definition}

Note that proper subvarieties of varieties always have measure $0$, so these notions of 
generic are consistent with one another.

\begin{example*}
	The set of $3 \times 3$ matrices forms a variety $V(S)$ with $S=\{0\}$. The property of having rank, 
or equivalently Kruskal rank, $3$ is generic on $V$, since matrices of rank at most $2$, including those of the form
\eqref{eq:Krank}, lie in a finite union of lower dimensional sets. This subvariety of exceptional matrices is 
defined by a single polynomial, the $3\times3$ determinant. 
\end{example*}

A fundamental tool for drawing conclusions that model parameters are generically identifiable 
is the following variant of a proposition in \cite{RS2012}, which we use repeatedly.

\begin{prop} \label{prop:genericProp} 
	Let $\Phi:U \to \mathbb C^n$ be an complex analytic map with $U$ an open subset of $\mathbb C^\ell$.
	Let $V$ be a variety in $\mathbb C^n$. Suppose $f\in I(V)$, and that there exists a point 
	$p_1=\Phi(u_1)$ with $f(p_1)\neq 0$. Then for generic points $u\in U$ or $u\in U\cap\mathbb R^n$, 
	the point $\Phi(u)$ lies off of $V$.
\end{prop}
\begin{proof} This follows from basic properties of complex analytic functions of many variables 
	(see, for instance, the text by \citet{Range1986}).
The function $f \circ\Phi$ is analytic, and not identically zero. Its zero set is therefore of 
measure zero, so for generic $u\in U$, $\Phi(u)$ lies off
$V(f)\supseteq V$. The real points in the zero set must similarly have measure zero.
\end{proof}

\subsection{Rank Propositions} For the proof of our main theorem,
the ranks and Kruskal ranks of some special matrices arising in the PM model are needed,
and we compile these rank computations here.  By giving these algebraic results in advance, 
the proof of Theorem \ref{thm:mainThm} can be presented more cleanly. 
Note that our arguments depend in part on some computations that were performed 
with the software \pari. As these computations were performed using exact integer arithmetic, 
they may be taken as valid proofs, up to the usual assumptions of correct programming and 
no hardware faults.

\smallskip

We begin by defining a particular structured matrix that can arise from particular
parameter choices for the PM model.

\begin{definition} \label{def:specialMarkovMatrix} 
	With $a_i\in \mathbb C$ for $i\in [\kappa]$, and $s=a_1+\dots+a_\kappa$,  
	let $M(a_1,\dots,a_\kappa)$ denote the $\kappa \times \kappa$ matrix 
	\begin{equation}\label{eq:specialMatrix}
	M(a_1,\dots,a_\kappa)=\begin{pmatrix}
	1+a_1-s&a_2&\cdots&a_{\kappa}\\
	a_1&1+a_2-s&\cdots&a_\kappa\\
	\vdots& \ddots& &\vdots \\
	a_1&a_2&\cdots&1+a_{\kappa}-s
	\end{pmatrix}.
	\end{equation}
\end{definition}

\begin{prop} \label{prop:rankTerminal} 
	For $\kappa=20$ and $m \leq 77$, let $M$ be a $m\kappa\times\kappa$ matrix formed by stacking $m\ge 1$  
	choices of matrices of the form  $M(a_1,\dots,a_\kappa)$.  Then 
	$M^{\otimes_r^{\ell}}$ has full row rank for  generic choices of the $a_i$ when $\ell\geq3$.
\end{prop}
\begin{proof}
We begin with the special case of $\ell=3$. An exact \pari calculation 
shows that for $m=77$ by picking distinct random integers for $a_1,\dots, a_\kappa$ 
for each of the $m$ blocks in $M$, we may find a point $p_1=M$ for which $M^{\otimes_r^3}$ has full row rank.  
By removing some of the blocks from this example if $m < 77$ we obtain a point $p_1$ for which 
$M^{\otimes_r^{3}}$ has full row rank for smaller $m$ as well.   
	
To show that full row rank is a generic condition when $\ell = 3$, fix $m\leq 77$, and  
observe that the map from  the space $\mathbb C^{m\kappa}$ of the $a_i$ to $M$ is analytic.  Since $p_1=M$ gives $M^{\otimes_r^{3}}$ full row rank, there is some
$m\kappa\times m\kappa$ minor $f$ of $M^{\otimes^3_r}$ which when viewed as a polynomial in the 
entries of $M$ has $f(p_1)\neq 0$. Taking $V=V(f)$, then Proposition \ref{prop:genericProp}
shows that generic choices of the $a_i$ give $f(M)\ne 0$ so $M^{\otimes^3_r}$ has rank $m\kappa$.	 
	
Now consider $\ell> 3$. Then 
$M^{\otimes_r^{\ell}}=M^{\otimes_r^{3}}\otimes_r M^{\otimes_r^{\ell-3}}$,
where $M^{\otimes_r^{3}} = \left ( \mu_{ij} \right)$ is a $m\kappa\times \kappa^3$ matrix and 
$M^{\otimes_r^{\ell-3}} = \left ( \alpha_{kl} \right)$ is a $m\kappa\times \kappa^{\ell-3}$ matrix.
Since $M^{\otimes_r^{3}}$ has full row rank $m\kappa$ for 
generic $M$, its rows are independent. But, with $v=m\kappa$,
$$
M^{\otimes_r^{3}}\otimes_r M^{\otimes_r^{\ell-3}}
=\begin{pmatrix}
\mu_{11}\alpha_{11}&\mu_{12}\alpha_{11}&\cdots&\mu_{1\kappa^3}\alpha_{11}&\cdots\cdots\\
\mu_{21}\alpha_{21}&\mu_{22}\alpha_{21}&\cdots&\mu_{2\kappa^3}\alpha_{21}&\cdots\cdots\\
\vdots& \ddots& &\vdots \\
\mu_{v1}\alpha_{v1}&\mu_{v2}\alpha_{v1}&\cdots&\mu_{v\kappa^3}\alpha_{v1}&\cdots\cdots
\end{pmatrix},
$$
so it is enough to know that the entries of  some single column of $M^{\otimes_r^{\ell-3}}$ 
are nonzero  and that $M^{\otimes_r^3}$ has independent rows to ensure $M^{\otimes_r^{\ell}}$ 
has independent rows.  But this is true for generic choices of parameters for $M$.
\end{proof}

The next proposition gives a lower bound on Kruskal row rank, valid for all $M^{\otimes_r^{\ell}}$.

\begin{prop} \label{prop:LowerBoundRankTerminal}\label{Prop3.19}
	For $\kappa\geq2$, let $M$ be a $m\kappa\times\kappa$ matrix formed by stacking $m \ge 1$ choices of 
	matrices of the form  $M(a_1,\dots,a_\kappa)$. For $\ell\geq 1$, $M^{\otimes_r^{\ell}}$ has Kruskal 
	row rank greater than or equal to $2$ for generic choices of the $a_i$.
\end{prop}

\begin{proof}
	Consider first the case that $\ell=1$. The matrices of Kruskal rank at most 1 form an 
	algebraic variety $V$. By Proposition \ref{prop:genericProp}, it is enough to find a single
	matrix $M$ not in $V$ to see that generically such matrices have Kruskal rank at least two.
	Choose $m\kappa$ distinct positive small numbers as the free entries $a_1,\dots, a_{\kappa}$ 
	in each block of $M$, so that the diagonal entries are the largest in the block. Then no two rows 
	within any block $M(a_1,\dots,a_\kappa)$ are multiples of each other, and no two rows of different 
	blocks are multiples either, since the $a_i$'s are distinct. Thus  $M$ has Kruskal rank greater than or equal to two.
	
	The case when $\ell>1$ follows by an argument similar to that at the end of the proof 
	of Proposition \ref{prop:rankTerminal}.
\end{proof}

The final propositions in this section involve generic ranks of stacked matrices formed 
by taking certain tensor products of matrices of the form above.

\begin{prop} \label{prop:rankWrongSplit}   
	Let $M$ be a $m\kappa^2\times \kappa^3$ matrix formed by stacking $m$ choices of matrices 
	of the form $M(a_1,\cdots, a_\kappa)^{\otimes^2_r}\otimes M(a_1,\cdots, a_\kappa)$. Then 
	for $\kappa=20$ and $m < 77$,  the matrix $M$ has rank greater than $m\kappa$ for generic choices of the $a_i$.
\end{prop}

\begin{proof}
	A \pari calculation shows that for some choice of random integers $a_i$, $M$ has 
	\begin{enumerate}
		\item full row rank $400 > m\kappa=20$, when $m=1$;
		\item full row rank $800 > m\kappa=40$, when $m=2$;
		\item rank $1180 > m\kappa=60$, when $m=3$; and 
		\item rank $1540 > m\kappa=80$, when $m=4$. \label{test}
	\end{enumerate}
	Furthermore, by \eqref{test}, for $m\ge 5$, there exists a matrix $M$ with 
	rank at least $1540 = 20 \times 77$ for some choice of $a_i$'s, since we may repeat some blocks.  Using Proposition \ref{prop:genericProp},
	the stated rank condition on $M$ is thus generic for all $m < 77$.
\end{proof}

\begin{prop} \label{prop:rankJD}   
	Let $M_1$ be of the form of $M$ in Proposition \ref{prop:rankWrongSplit}, and $M_2$ be formed by 
	stacking $m$ matrices of the form $M(a_1,\cdots, a_\kappa) \otimes M(a_1,\cdots, a_\kappa)^{\otimes^2_r}$. 
	Let $L$ be a $m\kappa^2\times m\kappa^2$ diagonal matrix with positive entries. Then for 
	$\kappa=20$ and $m<74$, $M_2^T L M_1$ has rank greater than $m\kappa$ for generic choices of the $a_i$.
\end{prop}

\begin{proof}
	
	Sylvester's rank inequality gives
	$$
	\mathrm{rank}(M_2^T L M_1)\geq \mathrm{rank}(M_2^T)+\mathrm{rank}(LM_1)-m\kappa^2.
	$$
	Since $M_1$ and $M_2$ differ only by row and column permutations, they have the same rank.
	Moreover, $\rank(L) = \rank(L M_1)$ since $L$ is a diagonal matrix with positive entries.
	Then, by Proposition \ref{prop:rankWrongSplit}, there is a choice of $a_i$'s so that 
	$M_2^T L M_1$ has rank at least
	\begin{enumerate}
		\item $400+400-400=400 > m \kappa = 20$, when $m=1$;
		\item $800+800-800=800 > m \kappa = 40$, when $m=2$;
		\item $1180+1180-1200=1160 > m \kappa = 60$, when $m=3$; and
		\item $1540+1540-1600=1480 > m \kappa = 80$, when $m=4$.
	\end{enumerate}
	The rank computation for $m = 4$ shows additionally that there exist choices of $a_i$ giving  
	$\rank ( M_2^T L M_1 ) = 1480$ for larger $m$ , since blocks can be repeated. But $1480 = 20 \times 74$ so by
	 Proposition  \ref{prop:genericProp}, generically then the rank
	must be greater than $ m\kappa$ for all $m< 74$.
\end{proof}

\section{Algebraic Aspects of the Profile Mixture Model} \label{sec:AlgPM}

Next we relate the algebraic definitions made in the previous section to phylogenetic models and the 
PM model in particular. We begin
by describing how a row tensor product of Markov matrices relates to parameters on a star tree.

\begin{definition} \label{def:starMatrixParam}    
	Let $A$ be a set of taxa on a star tree rooted at its internal node, with pendant edges 
	$e_1,\dots,e_{|A|}$ and associated Markov matrices $M^{e_i}$. Then 
	\begin{equation}\label{eq:matrixStarTree}
	M_{A}= M^{e_1}\otimes_r \dots\otimes_r  M^{e_{|A|}}.
	\end{equation}       
\end{definition}

For an $m$-class PM model on a star tree, the matrix $M_{A}$ is of size $m\kappa\times\kappa^{|A|}$.
Its entries are conditional probabilities of observing different $|A|$-tuples of states at the taxa in 
set $A$, given the state at the root. 

\medskip

Given a tree $T$ on taxa $X$, tripartitions and splits of $X$
can be associated to the topological structure of $T$. 
For instance, the tree of Figure \ref{fig:tree} displays a tripartition $A|B|C$ 
with $A=\{a,b,c\}, B=\{d,f\}, C=\{g,h\}$.  Formally, a tripartition 
$A|B|C$ is \emph{displayed} on a tree if there is some vertex $v$ of $T$ whose deletion results 
in three subtrees with  $A, B, C$ labeling their leaves.  Similarly,
if $A'=\{a,b,c\}$ and $B'=\{d,f,g,h\}$, then $X = A' \sqcup B'$, and $T$ \emph{displays} the split $A'|B'$
of $X$, since there is an edge $e$ whose deletion results in two subtrees 
with leaves labeled by $A'$ and $B'$.	

\begin{figure}[h]
	
	\begin{center}
		\begin{tikzpicture}[node distance=2cm]
		\coordinate[] (cint) {};
		\coordinate[left=0.5cm and 1cm of cint] (lint) {};
		\coordinate[above right=0.5cm and 1cm of cint] (rint) {};
		\coordinate[left=0.5cm and 0.8cm of lint] (c) {};
		\coordinate[below =0.9cm and 0.4cm of lint] (d) {};
		\coordinate[below left=0.5cm and 0.5cm of d] (j) {};
		\coordinate[below right=0.5cm and 0.5cm of d] (k) {};
		\coordinate[below right=0.9cm of cint] (a) {};
		\coordinate[above right=0.5cm and 0.5cm of rint] (b) {};
		\coordinate[below right=0.5cm and 0.5cm of rint] (e) {};
		\coordinate[below left=0.5cm and 0.5cm of c] (g) {};
		\coordinate[above left=0.5cm and 0.5cm of c] (i) {};

		\draw (d) -- (lint)  (lint) -- (cint)  (rint) -- (e)  (cint) -- (a) (rint) -- (b)  (c)-- (g)  (c)-- (i)   (cint)-- (rint) (lint) --(c)  (d) --(j)   (d) --(k);
			
		\begin{scope}[node distance=0.3cm]
		\node[left of=g] {$h$};
		\node[left of=i] {$g$};
		\node[right of=b] {$a$};
		\node[right of=e] {$b$};
		\node[below of=j] {$d$};
		\node[below of=k] {$f$};
		\node[right of=a] {$c$};
		\node[above right=-0.1 cm of lint, xshift=-0.15cm] {$v$};
		\node[above right=-0.1 cm of lint, xshift=0.3cm] {$e$};
		\end{scope} 			
		\end{tikzpicture}
	\end{center}
	\caption{A tree displaying the tripartition $A|B|C$ and the split $A|B\cup C$,
		where $A=\{a,b,c\}, B=\{d,f\}, C=\{g,h\}$.} \label{fig:tree}   
\end{figure}
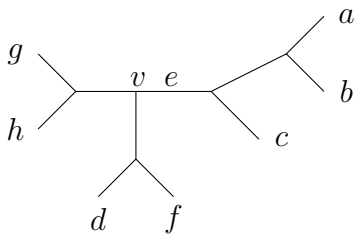	

When a tree $T$ displays a tripartition of a set of taxa, then the flattening of a joint 
distribution corresponding to that tripartition can be expressed using the $3$-way matrix 
product of certain matrices built from model parameters.

\begin{lemma} \label{lem:vertexParams}   
	Suppose $T$ is a tree on a set of taxa $X$ rooted at an internal vertex $v$ and that $T$ displays
	the tripartition $A|B|C$ associated to $v$. 
	Let $P$ be a probability distribution for a 
	Markov model $\mathcal M$ 
	on $T$ with 
	$\ell$ states at the internal nodes. 
	Then there exist matrices $\overline{M}_A$, $\overline{M}_B$, $\overline{M}_C$ constructed
	from model parameters for $\mathcal M$, each with $\ell$ rows, such that
	$$
	\operatorname{Flat}_{A|B|C}(P)=[\overline{M}_A,\overline{M}_B,\overline{M}_C].
	$$
\end{lemma}

\begin{proof}
	From the parameters on $T$ we may define Markov matrices $M_{A}, M_{B}, M_{C}$ 
	whose entries are conditional probabilities of states at 
	the leaves in each set $A,B,C$, given the state at $v$. Let $\boldsymbol \pi$ 
	be the state distribution at $v$. Then 
	$$
	\operatorname{Flat}_{A|B|C}(P)=[\boldsymbol \pi; M_{A}, M_{B}, M_{C}]=[\overline{M}_{A}, \overline{M}_{B}, \overline{M}_{C}],
	$$   
where $\overline{M}_{A}=\diag(\boldsymbol \pi)M_A$, $\overline{M}_B = M_B$, and $\overline{M}_C = M_C$.
\end{proof}

For establishing generic properties of the PM model, we will often consider the particular choice of the exchangabilities given by the matrix $R=\mathds 1$
whose entries are all 1. This is in essence the CAT-F81 model \citep{LP04,QGL2008}, with the number of profiles some fixed $m$. For this $R$, a Markov matrix
has the form given in equation \eqref{eq:specialMatrix} of Definition \ref{def:specialMarkovMatrix}.

\begin{lemma} \label{lem:matForm} 
	Consider the PM model $PM(T,\kappa,m)$ with $R = \mathds{1}$, and let $e$ be a branch of $T$ of length 1.  
	Then for a single class $c$ with profile $\boldsymbol \pi$ and rate $r \ge 0$, 
	the Markov matrix $M^e_c = \exp(Q_c r)$ for $e$ is of the form
	$M(a_1, \dots, a_\kappa)$ of Definition \ref{def:specialMarkovMatrix}, with $a_i=\pi_i(1-e^{-r})\geq 0$ and $s=\sum_{i=1}^\kappa a_i$ 
	satisfying $0\le s< 1$.	
	
	Conversely, any $\kappa \times \kappa$ Markov matrix of the form $M = M(a_1,\dots,a_\kappa)$ 
	with  $a_j \ge 0$ and $0 \le s < 1$ 
	comes from a choice of parameters for one class of the PM model with $R = \mathds{1}$ on an edge of length 1. 
	
	Provided $s\ne 0$ (equivalently $r\ne 0$), this correspondence is one-to-one.
	\end{lemma} 
\begin{proof}
	The first statement follows by direct computation: With $\mathbf e_j$ the standard basis vectors, $Q_c=R\diag(\boldsymbol \pi)-I$ has right eigenvectors  $-\pi_j\mathbf e_1+\pi_1\mathbf e_j$ with eigenvalues  $-1$ for $2\le j\le \kappa$, and eigenvector $\sum_{j=1}^\kappa \mathbf e_j$ with eigenvalue $0$. 
	
	For the converse, 
	since $0\le s< 1$, there is a unique $r\ge 0$ such that $s=1-e^{-r}$. If $s>0$, 
	let $\pi_j={a_j}/{s}$ for $j=1,\cdots,\kappa$,
	and $\boldsymbol \pi = (\pi_j)$. Then $\sum_{j=1}^{\kappa}\pi_j=1$, 
	and $a_j=\pi_j(1-e^{-r})$. With these choices
	$Q=R \diag ( \boldsymbol \pi)-I$, and $M=\exp(rQ)$.
	If $s=0$, then all the $a_j$ are zero, and $M$ is the identity matrix.  
	Take $r=0$ and $\boldsymbol \pi$ arbitrary.  Then $M = \exp( 0Q )$. 
\end{proof}

\section{Identifiability of Parameters for the Profile Mixture Model}\label{sec:MainThm}

With preliminaries completed, we now turn to establishing our main result, on generic parameter 
identifiability for the PM model. 
 The first step is to understand that the ranks of matrix flattenings of a model distribution 
 are affected by whether the associated split is, or is not, displayed on the tree $T$.
 
\begin{prop} \label{prop:splitRank}
	Let $T$ be an $n$-taxon tree on $X$ and $P$ a distribution from the model PM= PM $(T, \kappa, m)$ with $\kappa=20$ and $m<74$. Suppose that $A|B$ is a split of $X$ with $|A|,|B|\geq 3$.  
	\begin{enumerate}
		\item[(1)] If $A|B$ is displayed on $T$, then $\operatorname{Flat}_{A|B}(P)$ has rank at most $m\kappa$;
		\item[(2)] If $A|B$ is not displayed on $T$, then $\operatorname{Flat}_{A|B}(P)$ generically has rank greater than $m\kappa$.  
	\end{enumerate}
\end{prop} 

Before beginning the proof, we present a simplified example to illustrate how
the matrix rank of flattenings of joint distributions from Markov models
on trees carries information about the absence/presence of an internal edges
on $T$. 

\begin{example*}  Consider a single-class 2-state Markov model
	on the $4$-taxon tree shown in Figure \ref{fig:4taxTree}.  A special case of this model is $PM(T,2,1)$.  
	 The joint distribution
	of states at the leave of $T$ is the $2 \times 2 \times 2 \times 2$ array $P$,
	with entries $p_{ijkl}$ indexed by leaves in the order $a, b, c, d$.
	
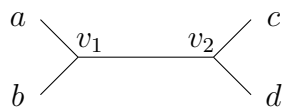
\begin{figure}[h] 
	\centering
	\begin{tikzpicture}[x=1.5cm, y=1.5cm]
	\coordinate[] (lsplit) {};
	\coordinate[above left=0.5cm and 0.5cm of lsplit] (esplit) {};
	\coordinate[below left=0.5cm and 0.5cm of lsplit] (d) {};
	\coordinate[right=1.5cm and 1.8cm  of lsplit] (rsplit) {};
	\coordinate[above right=0.5cm and 0.5cm of rsplit] (b) {};
	\coordinate[below right=0.5cm and 0.5cm of rsplit] (c) {};

	\foreach \a/\b in {lsplit/rsplit,d/lsplit,b/rsplit,c/rsplit,esplit/lsplit}
	\draw (\a) -- (\b);
	
	\begin{scope}[node distance=0.3cm]
	\node[left of=esplit] {$a$};
	\node[right of=b] {$c$};
	\node[right of=c] {$d$};
	\node[left of=d] {$b$};
	\node[above right=-0.1 cm of lsplit, xshift=-0.1cm] {$v_1$};
	\node[above left=-0.1 cm of rsplit, xshift=0.1cm] {$v_2$};
	\end{scope}  
	\end{tikzpicture}
	\caption{A 4-taxon tree with split $\{a,b\}|\{c,d\}$.} \label{fig:4taxTree}  
\end{figure}

With $A=\{a,b\}$ and $B=\{c,d\}$, 
the rows and columns of $\operatorname{Flat}_{A|B}(P)$ are indexed by elements of
$[2] \times [2]$.  For example, the $\left ( (1,2), \, (1,1) \right)$ entry is $p_{1211}$.
In contrast,  if $A' = \{a, c\}$ and $B' = \{b, d\}$, the flattening $\operatorname{Flat}_{A'|B'}(P)$ has 
$\left ( (1,2), \, (1,1) \right)$-entry
is $p_{1121}$.

Now suppose that the terminal edges of 
$T$ have length $0$, so that the states at $a$ and $b$ must agree, as must those
at $c$ and $d$, since no substitutions occur on terminal edges. 
Then the matrix $\operatorname{Flat}_{A|B}(P)$ arises from the joint distribution 
of states at the internal nodes $v_1$ and $v_2$, and its only non-zero entries are $p_{iijj}$. Thus 
the matrix flattening for the split $A|B$ displayed by $T$ has form
$$\operatorname{Flat}_{A|B}(P)=
\bordermatrix{
	&(1,1)&(1,2)&(2,1)&(2,2)\cr
	(1,1) & p_{1111} &0&0&p_{1122} \cr
	(1,2) & 0&0&0&0 \cr
	(2,1) & 0&0&0&0 \cr
	(2,2) & p_{2211}&0&0&p_{2222}
},
$$
with rank at most $2=m\kappa$. 

In contrast, the flattening for the split $A'|B'$ not displayed on $T$ has form
$$
\operatorname{Flat}_{A'|B'}(P)=
\bordermatrix{
	&(1,1)&(1,2)&(2,1)&(2,2)\cr
	(1,1) & p_{1111} &0&0&0 \cr
	(1,2) & 0&p_{1122}&0&0 \cr
	(2,1) & 0&0&p_{2211}&0 \cr
	(2,2) & 0&0&0&p_{2222}
},
$$
which generically has rank $4=(m\kappa)^2 > m\kappa$. 

If the terminal edges of $T$ are of positive length, then
the resulting joint distribution $P$ can be obtained by a simple
and generically rank-preserving linear action on the rows and columns of the flattenings above.
Thus, flattenings
respecting the topology of $T$ generically have rank $m\kappa$
while those that do not generically have larger rank.

\end{example*}

\begin{proof}[Proof of Proposition \ref{prop:splitRank}]
	
	To show claim $(1)$, suppose the split $A|B$ is displayed on $T$ with 
	associated edge $e=(v_A,v_B)$.  Let $M_A$ be the $m\kappa\times\kappa^{|A|}$ matrix 
	and $M_B$ the $m\kappa\times\kappa^{|B|}$ matrix
	giving the conditional probabilities
	of jointly observing states at $A$ and $B$, conditioned on states at $v_A$ and $v_B$ respectively.
	Then, by rooting the tree at $v_A$ and letting $M^e$ denote the $m\kappa \times m\kappa$
	Markov matrix associated to $e$, the joint distribution of $(v_A, v_B)$ is 
	$\diag(\boldsymbol{\Pi}) M^e$ and it follows that
	$$
	\operatorname{Flat}_{A|B}(P)=M_A^T \diag(\boldsymbol{\Pi}) M^e M_B.
	$$
	Since $\rank(M^e ) \le m\kappa$, it follows that $\operatorname{Flat}_{A|B}(P)$ has rank at most $m\kappa$. 
	
	For claim (2), suppose now $A|B$ is not displayed on $T$.
	Let $V$ be the variety of matrices of size $\kappa^{\vert A \vert} \times \kappa^{\vert B \vert}$
	with rank at most $m\kappa$, defined by the set of all $(m\kappa+1)\times (m\kappa+1)$ minors. 
	By Proposition \ref{prop:genericProp}, it suffices to find a single choice of $PM(T,\kappa,m)$ parameters 
	that produces a point off $V$, as the parameterization extends to a complex analytic function.
	
	Since $T$ does not display $A|B$, by Theorem 3.8.6 of \cite{SempleSteel}, there is an edge $e=(v_1,v_2)$ 
	of $T$ with associated split $C|D$ such that $A'=A\cap C$, $A''=A\cap D$, $B'=B\cap C$, $B''=B\cap D$ 
	are all non-empty.  To find the needed choice of parameters, fix all internal edges of $T$ except 
	$e$ to have length $0$, so the Markov matrices on these edges are $I$, and fix the edge lengths of all 
	terminal edges and $e$ to be $1$.  See Figure \ref{fig:treeSplitNotDisplayed}. 
	Take $R = \mathds 1$ and mixing weights $w_i=1/m$ to be uniform.
	Values for the parameters $\boldsymbol{\pi}_i, r_i$ will be specified later in the argument. 
For this choice of parameters, $T$ is formed by joining two star trees at the ends of $e$.

	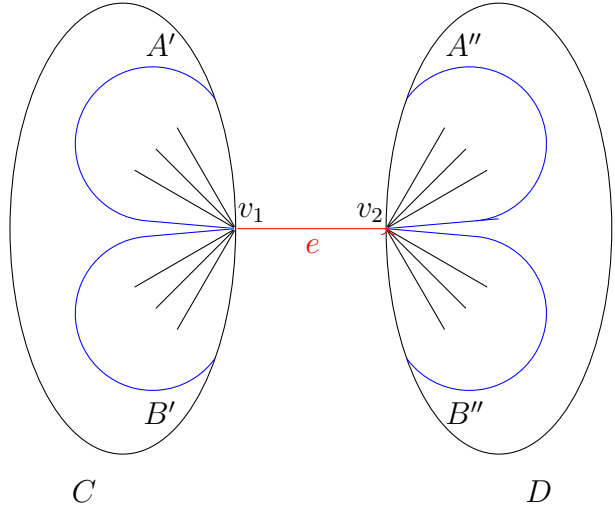
\begin{figure}[h]
		\centering
		\begin{tikzpicture}
		\pgfmathsetmacro{\is}{10}
		\draw (-2.5,0) circle (1.5cm and 3cm);
		\node[] (e0) {};
		\node[left of=e0,inner sep=0cm] (v1) {} [grow=north west,sibling distance=0.4cm]
		child
		child
		child
		;
		\node[left of=e0,inner sep=0] (v1t) {} [grow=south west,sibling distance=0.4cm]
		child
		child
		child
		;
		\coordinate (v1a1) at ($(v1)+(175:\tikzleveldistance)$);
		\node[right=1.5cm of v1, xshift=1.5cm] (v2) {};

		\begin{scope}
		\clip (v1) -- (-4,0) arc (180:30:1.5cm and 3cm) -- (v1); 
		\begin{scope}[rotate=-5]
		\coordinate (v1a2) at ($(v1)+0.8*(-\tikzleveldistance,0)$);
		
		\draw[blue] let \p1=($(v1-2)-(v1a1)$) in (v1a2) arc (270:0:{veclen(\x1,\y1)});
		\end{scope}
		\draw[blue] (v1) -- (v1a2);
		\end{scope}

		\coordinate (v1tb1) at ($(v1t)+(185:\tikzleveldistance)$);
		
		\begin{scope}
		\clip (v1t) -- (-4,0) arc (180:360:1.5cm and 3cm) -- (v1t); 
		\begin{scope}[rotate=5]
		\coordinate (v1tb2) at ($(v1t)+0.8*(-\tikzleveldistance,0)$);
		
		\draw[blue] let \p1=($(v1t-2)-(v1tb1)$) in (v1tb2) arc (90:360:{veclen(\x1,\y1)});
		\end{scope}
		\draw[blue] (v1t) -- (v1tb2);
		\end{scope}

		\draw (2.5,0) circle (1.5cm and 3cm);
		
		\node[right of=e0,inner sep=0cm] (v2) {} [grow=north east,sibling distance=0.4cm]
		child
		child
		child
		;
		\node[right of=e0,inner sep=0] (v2t) {} [grow=south east,sibling distance=0.4cm]
		child
		child
		child
		;
		\coordinate (v2a1) at ($(v2)+(5:\tikzleveldistance)$);
		
		\begin{scope}
		\clip (v2) -- (4,0) arc (0:150:1.5cm and 3cm) -- (v2); 
		\begin{scope}[rotate=5]
		\coordinate (v2a2) at ($(v2)+0.8*(\tikzleveldistance,0)$);
		
		\draw[blue] let \p2=($(v2-2)-(v2a1)$) in (v2a2) arc (-90:180:{veclen(\x2,\y2)});
		\end{scope}
		\draw[blue] (v2) -- (v2a1);
		
		\end{scope}

		\coordinate (v2tb1) at ($(v2t)+(-5:\tikzleveldistance)$);
		
		\begin{scope}
		\clip (v2t) -- (4,0) arc (360:180:1.5cm and 3cm) -- (v2t); 
		\begin{scope}[rotate=-5]
		\coordinate (v2tb2) at ($(v2t)+0.8*(\tikzleveldistance,0)$);
		
		\draw[blue] let \p1=($(v2tb1)-(v2t-2)$) in (v2tb2) arc (90:-180:{veclen(\x1,\y1)});
		\end{scope}
		\draw[blue] (v2t) -- (v2tb2);
		\end{scope}

		\begin{scope}[node distance=0.3cm]
		\node[above right of=v1] {$v_1$};
		\node[above left of=v2] {$v_2$};
		\node[above left=3cm of v1, xshift=1.5cm]  {$A'$};
		\node[below left=3cm of v1, xshift=1.5cm]  {$B'$};
		\node[below left=4.5cm of v1, xshift=1.5cm]  {$C$};
		\node[above right=3cm of v2, xshift=-1.5cm]  {$A''$};
		\node[below right=3cm of v2, xshift=-1.5cm]  {$B''$};
		\node[below right=4.5cm of v2, xshift=-1.5cm]  {$D$};
		\end{scope}
		\draw[->,red] (v1) -| node[near start,below] {$e$} (v2);
		\end{tikzpicture}
		\caption{A tree $T$ which does not display the split $A|B$, but displays the split $C|D$ such that $A'=A\cap C$, $A''=A\cap D$, $B'=B\cap C$, $B''=B\cap D$ are all non-empty.} \label{fig:treeSplitNotDisplayed} 
	\end{figure}
	
	Taking $r=v_1$ to be the root of $T$, let $K=\text{diag}(\boldsymbol \Pi) M^{e}$ be the $m\kappa\times m\kappa$ 
	block diagonal matrix which is the joint distribution of classes and states at $v_1$ and $v_2$. 
	The probabilities of observing states $\boldsymbol{i}$, $\boldsymbol{j}$, $\boldsymbol{k}$, $\boldsymbol{l}$ at 
	leaves in $A'$, $B'$, $A''$, $B''$ respectively, $P( \boldsymbol{i}, \boldsymbol{j}, \boldsymbol{k}, \boldsymbol{l})$, 
	are the entries of a $\kappa^{|A'|}\times \kappa^{|B'|}\times\kappa^{|A''|}\times\kappa^{|B''|}$ tensor. 
	
	Define a $m\kappa\times m\kappa\times m\kappa\times m\kappa$ tensor $\overline Q$,
	
	$$
	\overline Q( i,j, k, l)=
	\begin{cases}
	K( i, k)& i=j, \,k=l,\\
	0 & \text{otherwise}.
	\end{cases}
	$$ 
	 The tensor $\overline Q$ is the joint distribution of states at the leaves of the tree $T$ 
	 of Figure \ref{fig:treeSplitNotDisplayed} 
	 when terminal edges have length zero and $A', B', A'', B''$ are single taxa. 
	 Indeed, since 
	 $A|B$ is not displayed on $T$, the matrix $\widehat{Q} = \operatorname{Flat}_{A|B}(\overline Q)$ is 
	 $(m\kappa)^2\times(m\kappa)^2$ with entries 
	$$
	\widehat{Q}\left((i, j),(k, l)\right)=\overline Q(i, k, j, l).
	$$
	Since $K$ is block diagonal, $\widehat Q$ has at most $m\kappa^2$
	nonzero entries, all appearing on the diagonal, and $\widehat Q$ is
	generically of rank $m\kappa^2$. 
	
	To see that in the general case $\operatorname{Flat}_{A|B}(P)$ has a similar structure, let 
	$N_A=M_{A'} \otimes M_{A''}$ and $N_B=M_{B'} \otimes M_{B''}$ where $M_{A'}, M_{A''}, M_{B'}, M_{B''}$ 
	are given as in equation \eqref{eq:matrixStarTree} of Definition \ref{def:starMatrixParam}. Then
	\begin{equation} \label{eq:flatWrongSplit} 
	\operatorname{Flat}_{A|B}(P)=N_A^T \, \widehat{Q} \, N_B.
	\end{equation}

	\begin{figure}[h]
		\centering
		(a) \begin{tikzpicture}[x=1.5cm, y=1.5cm]
		\coordinate[] (lsplit) {};
		\coordinate[above left=0.9cm and 0.5cm of lsplit] (esplit) {};
		\coordinate[above left=0.7cm and 0.7cm of lsplit] (g) {};
		\coordinate[below left=0.9cm and 0.5cm of lsplit] (d) {};
		\coordinate[below left=0.7cm and 0.7cm of lsplit] (h) {};
		\coordinate[right=1.5cm and 1.8cm  of lsplit] (rsplit) {};
		\coordinate[above right=0.9cm and 0.5cm of rsplit] (b) {};
		\coordinate[below right=0.9cm and 0.5cm of rsplit] (c) {};
		
		\foreach \a/\b in {lsplit/rsplit,d/lsplit,b/rsplit,c/rsplit,esplit/lsplit,lsplit/g,lsplit/h}
		\draw (\a) -- (\b);
		
		\begin{scope}[node distance=0.5cm]
		\node[above right=-0.1 cm of lsplit, xshift=-0.1cm] {$v_1$};
		\node[above left=-0.1 cm of rsplit, xshift=0.1cm] {$v_2$};
		\node[below right=0.1 cm of lsplit, xshift=0.7cm] {$e$};
			\node[above right=-0.1 cm of b, xshift=0.1cm, yshift=-0.1cm] {$A''$};
			\node[below right=-0.1 cm of c, xshift=0.1cm] {$B''$};
			\node[above left=0.1 cm of g, xshift=0.1cm, ] {$A'$};
			\node[below left=0.1 cm of d, xshift=0.1cm, yshift=0.1cm] {$B'$};
		\end{scope}  
		\end{tikzpicture}
		\hskip 2.0cm
		(b)\begin{tikzpicture}[x=1.5cm, y=1.5cm]
		\coordinate[] (lsplit) {};
		\coordinate[above left=0.9cm and 0.5cm of lsplit] (esplit) {};
		\coordinate[above left=0.7cm and 0.7cm of lsplit] (g) {};
		\coordinate[below left=0.9cm and 0.5cm of lsplit] (d) {};
		\coordinate[right=1.5cm and 1.8cm  of lsplit] (rsplit) {};
		\coordinate[above right=0.9cm and 0.5cm of rsplit] (b) {};
		\coordinate[below right=0.9cm and 0.5cm of rsplit] (c) {};
		\coordinate[below right=0.7cm and 0.7cm of rsplit] (h) {};

		\foreach \a/\b in {lsplit/rsplit,d/lsplit,b/rsplit,c/rsplit,esplit/lsplit,lsplit/g,rsplit/h}
		\draw (\a) -- (\b);
		
		\begin{scope}[node distance=0.5cm]
		\node[above right=-0.1 cm of lsplit, xshift=-0.1cm] {$v_1$};
		\node[above left=-0.1 cm of rsplit, xshift=0.1cm] {$v_2$};
		\node[below right=0.1 cm of lsplit, xshift=0.7cm] {$e$};
			\node[above right=-0.1 cm of b, xshift=0.1cm,yshift=-0.1cm] {$A''$};
		\node[below right=-0.1 cm of c, xshift=0.1cm] {$B''$};
		\node[above left=0.1 cm of g, xshift=0.1cm] {$A'$};
		\node[below left=0.1 cm of d, xshift=0.1cm, yshift=0.1cm] {$B'$};
		\end{scope}  
		\end{tikzpicture}
		\caption{Trees with (a) $|A'|=|B'|=2$ and $|A''|=|B''|=1$, and (b)  $|A'|=|B''|=2$ and $|A''|=|B'|=1$		
		.} \label{fig:treeAB}  
	\end{figure}
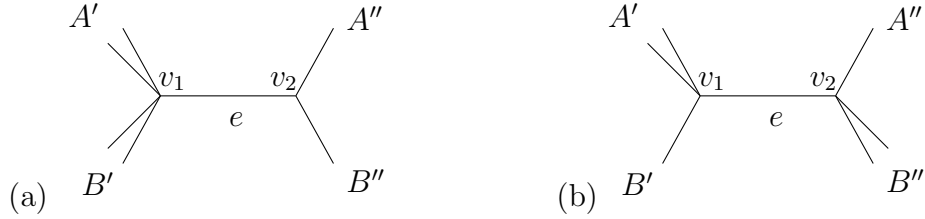
	
	We  now establish that claim $(2)$ holds when $|A|=|B|=3$, so the tree is one of those shown in Figure \ref{fig:treeAB}.
	 Suppose first that $|A'|=|B'|=2$ and 
       $|A''|=|B''|=1$, as shown for tree (a) of the figure. In this case $N_A=N_B$. 
	Since $\widehat Q$ is diagonal with at most $m\kappa^2$ non-zero entries due to the block structure of $K$, 
	in equation \eqref{eq:flatWrongSplit} we can replace 
	$\widehat{Q}$ by a diagonal $m\kappa^2 \times m\kappa^2$ matrix $Q$ by eliminating
	zero rows and columns.  To do this, we must also replace $N_A = N_B$ with an $m\kappa^2\times\kappa^3$ matrix
	$N$ formed by taking tensor products of the individual class components
of $M_{A'} = M^{\otimes^2_r}$ and $M_{A^{''}} = M$ and then restacking.
To be concrete, for class $c$ the Markov matrix for a terminal edge is $M^c = M(a^c_1, \dots, a^c_\kappa)$
by Lemma \ref{lem:matForm}, and $N$ is formed by stacking $m$ matrices $(M^c)^{\otimes^2_r} \otimes M^c$.
	
Since $Q$ is diagonal with generically positive entries, using equation \eqref{eq:flatWrongSplit} we have that
	$$
	\operatorname{Flat}_{A|B}(P)= \left(N^T Q^{1/2} \right) \, \left( Q^{1/2}  N \right) = \Lambda^T \, \Lambda, 
	$$
where $\Lambda= Q^{1/2} N$.  By the singular value decomposition, it follows  that 
	$$
	\rank(\Lambda^T \Lambda)=\rank(\Lambda)=\rank(N).
	$$
The \pari calculation presented in Proposition \ref{prop:rankWrongSplit}, together with Proposition \ref{prop:genericProp} 
show that $\rank(N)> m\kappa$ generically, and thus for generic $\boldsymbol \pi_i$ and $r_i$
it follows that $\rank(\operatorname{Flat}_{A|B}(P))>m\kappa$.
	
	Now continuing with $|A|=|B|=3$ suppose that $|A'|=|B''|=2$ and $|A''|=|B'|=1$, as shown by Figure \ref{fig:treeAB}(b). The previous 
	argument fails for this tree 
	because now $N_A\neq N_B$, as the tensor products defining these matrices, are taken 
	in different orders.  However, a more complicated \pari calculation, presented as 
	Proposition \ref{prop:rankJD}, shows that $\operatorname{Flat}_{A|B}(P)$ generically has 
	rank greater than $m\kappa$ in this case.  
	
	Finally, for the general case of $|A|,|B|\geq3$, take $\widehat{A}$ to be a 
	$3$-element subset of $A$ with at least one element from $A'$ and one from $A''$, and similarly 
	take $\widehat{B}$ to be a $3$-element subset of $B$ with at least one element from $B'$ and from 
	$B''$. Let $\widehat{P}$ be the probability distribution for the taxa $\widehat{A}\cup\widehat{B}$. 
	Since the row indices of $\operatorname{Flat}_{A|B}(P)$ depend on the states at the taxa in $A$ and 
	the column indices depend on the states at the taxa in $B$, marginalizing over all possible states 
	for the taxa in $A$ which are not in $\widehat{A}$, and similarly for $B$,  gives the matrix
 $\operatorname{Flat}_{\widehat{A}|\widehat{B}}(\widehat{P})$. There exist matrices, $J_1, J_2$ which 
	perform this marginalization on $\operatorname{Flat}_{A|B}(P)$, 
	$$J_1 \operatorname{Flat}_{A|B}(P) J_2= \operatorname{Flat}_{\widehat{A}|\widehat{B}}(\widehat{P}).
	$$
	Since $\operatorname{Flat}_{\widehat{A}|\widehat{B}}(\widehat{P})$ generically has rank greater than 
	$m\kappa$ and $\operatorname{Flat}_{A|B}(P)$ has rank greater than 
	or equal to $\operatorname{Flat}_{\widehat{A}|\widehat{B}}(\widehat{P})$ by this equation, 
	it follows that $\operatorname{Flat}_{A|B}(P)$ generically has rank greater than $m\kappa$.  
\end{proof}

As a consequence of Proposition \ref{prop:splitRank}, from a distribution $P$ computed from generic 
PM model parameters we can identify every edge in the tree for which there are at least three taxa on either side, by computing ranks of flattenings
of $P$.
In the following, 
we see that Proposition \ref{prop:splitRank} also helps to identify at least one tripartition on the tree.

\begin{prop}   \label{prop:IDtripartion} 
	Let $T$ be an $n$-taxon tree on $X$ with $n \ge 9$, and $P$ a joint distribution from 
	generic parameters for the model $PM(T, \kappa, m)$ with $\kappa=20$ and $m < 74$.  Then there is at least one 
	tripartition $A|B|C$ displayed on $T$, with $|A|,|B|\ge 3$, which can be identified from $P$.
\end{prop}

\begin{proof}
	By Lemma $4.8$ of \citet{RS2012}, every unrooted binary tree $T$ with $n\geq 3$ has an 
	internal vertex $v$ which induces a tripartition $A|B|C$ such that two of the three components 
	contain at least $\lceil n/4 \rceil$ leaves of $T$. 
	
	The two edges incident to $v$ that correspond to subsets of $X$ with at least $\lceil n/4 \rceil$ 
	leaves are generically identifiable by Proposition \ref{prop:splitRank}, since for $n\geq 9$, 
	$\lceil n/4 \rceil\geq 3$. If the third edge incident to $v$ has $3$ or more taxa in its component, 
	it also can be identified. Thus, it remains to establish that the third edge incident to $v$ can be 
	identified when the number of taxa in its component is $1$ or $2$. Examples of such trees are illustrated for 
	$n=9$ in Figure \ref{Fig5.7}.
	
	\begin{figure}[h]
		\centering
			$(a)$ \includegraphics[width=2.5in]{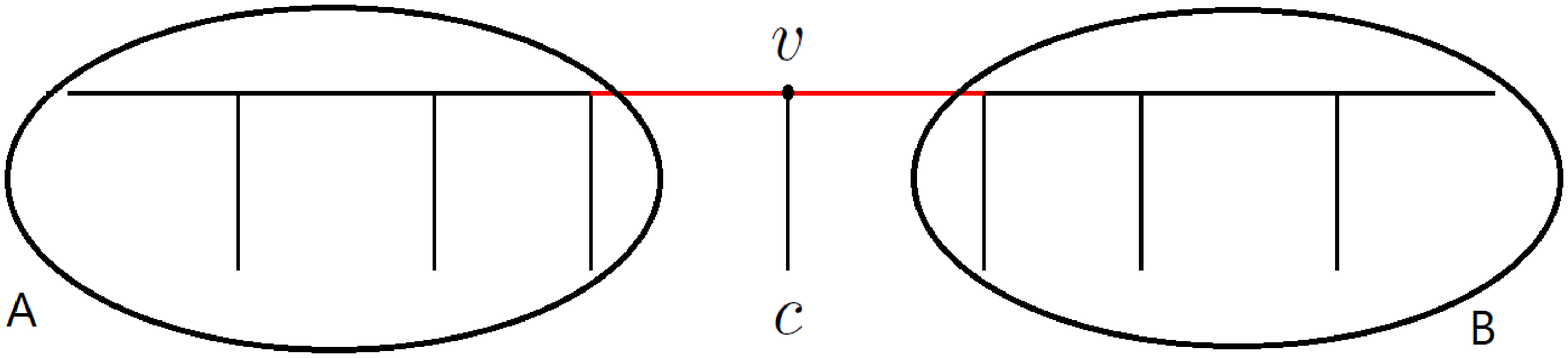}
		\hspace*{0.1in}
		$(b)$ \includegraphics[width=2in]{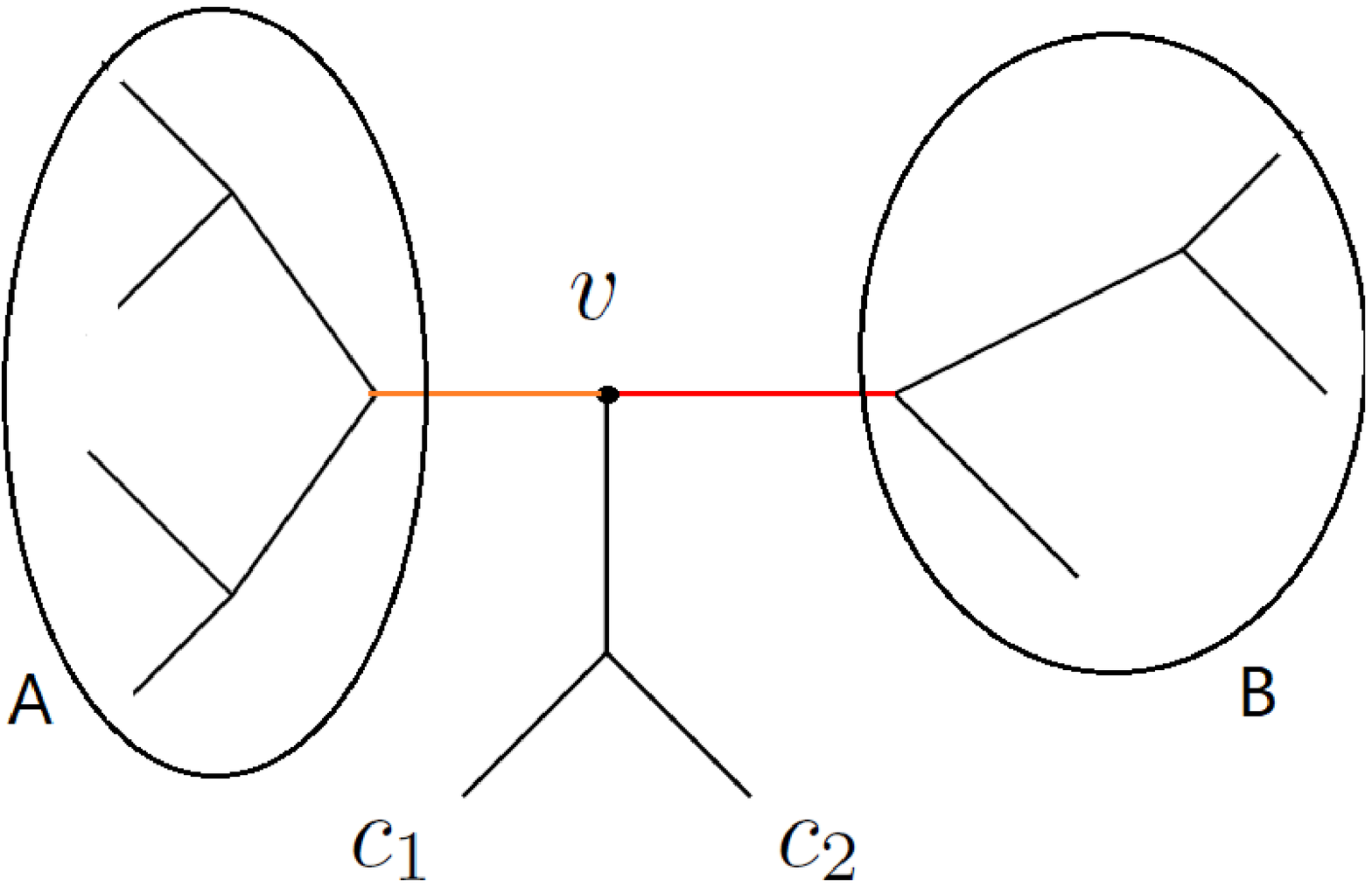}
		\caption{Examples of 9-taxon trees with internal vertex $v$ inducing $A|B|C$ with $|A|,|B|\geq 3$ and $|C|=1$ or $2$.} \label{Fig5.7}
	\end{figure}
	
	If the third component has only one leaf, as in Figure \ref{Fig5.7}(a), the two bipartitions 
	$A\cup\{c\}|B$ and $A|B\cup\{c\}$ are identifiable by Proposition \ref{prop:splitRank}. Together this implies that  
	the tripartition induced by $v$ is $A|B|\{c\}$. If the third component has two leaves as in  
	Figure \ref{Fig5.7} (b), the two splits $A\cup\{c_1, c_2\}|B$ and $A|B\cup\{c_1, c_2\}$ are identifiable, 
	but $A\cup\{c_1\}|B\cup\{c_2\}$ and $A\cup\{c_2\}|B\cup\{c_1\}$ are not displayed on $T$, and that can be 
	detected by Proposition \ref{prop:splitRank}. This implies the tripartition $A|B|\{c_1, c_2\}$ is on the tree.
\end{proof}	

With a tripartition on the tree identifiable by the preceding proposition,
we prepare to apply Kruskal's Theorem.
Letting $P$ be a joint distribution from $PM(T,\kappa,m)$, pick an internal
vertex $v$ of $T$ inducing such a tripartition $A|B|C$.
Then  by Lemma \ref{lem:vertexParams} 
$$
\operatorname{Flat}_{A|B|C}(P)=[\boldsymbol{\pi}; {M}_A, {M}_B, {M}_C]=[\overline{M}_A, {M}_B, {M}_C],
$$
where $\overline{M}_A=\diag({\boldsymbol \Pi})M_A$. 
Provided the Kruskal ranks of the matrices 
$\overline{M}_A, {M}_B, {M}_C$ are large enough, at least generically,
Kruskal's theorem can be applied.  The next three lemmas establish this.

\begin{lemma} \label{lem:rankStarMatA} \label{Lem30} 
    Consider the model $PM(T, 20, m)$ with
    $m \le 77$.
    If $\ell \ge 3$, then the $\ell^{th}$ row tensor power of the 
	$m\kappa \times \kappa$ Markov matrix associated to a terminal edge of $T$ 
	has full row rank for generic parameters.
\end{lemma}

\begin{proof} Using Proposition \ref{prop:genericProp}, it is enough to show there is 
	a single choice of parameters for which the tensor power has full row rank. 
Let $R = \mathds 1$, and take the terminal branch lengths to be $1$. Then by Lemma \ref{lem:matForm}  the 
Markov matrix $M_e$ on a terminal edge has the form of stacked matrices of the form $M(a_1,\dots, a_\kappa)$. 
By the \pari calculation of Proposition \ref{prop:rankTerminal}, for generic choices of the other 
parameters, $M_e^{\otimes_r^\ell}$, $\ell\ge 3$, has full row rank.
\end{proof}

Using Proposition \ref{Prop3.19} in a similar argument we obtain the following.

\begin{lemma} \label{lem:rankStarMatB} 
	Consider the model $PM(T, \kappa, m)$ with $\kappa \ge 2$ and
	$m \ge 1$.
	Then for $\ell \ge 1$, the $\ell^{th}$ row tensor  power of the 
	$m\kappa \times \kappa$ Markov matrix associated to a terminal edge of $T$  
	generically has Kruskal rank at least 2.
\end{lemma}

\begin{lemma} \label{prop:ranksBigEnough} 
	For a distribution from the model $PM(T,\kappa,m)$ with $\kappa=20$ and $m\leq 77$, let 
	$\overline{M}_A, {M}_B, {M}_C$ be the matrices described above. 
	If $|A|,|B|\geq 3$, and $|C|\geq 1$, then generically $\overline{M}_A$, ${M}_B$ have full 
	Kruskal rank and ${M}_C$ has Kruskal rank at least $2$.
\end{lemma}

\begin{proof}
	Using Proposition \ref{prop:genericProp}, we need only show there is a single choice of parameters for 
	which these rank claims hold.  Set
	all internal branch lengths 0 and all terminal branch
	lengths $1$, so that $T$ is a star tree rooted at the central node $v$. 
Then by Lemma \ref{lem:rankStarMatA}, since $|A|,|B|\geq 3$ for generic choices of the profiles $\boldsymbol{\pi}_i$ the matrices ${M}_A$ (and therefore $\overline{M}_A$) and ${M}_B$ have full row rank and 
therefore full Kruskal rank. Also by Lemma \ref{lem:rankStarMatB}, ${M}_C$ has Kruskal rank at least $2$.
\end{proof}

We add the last ingredient before the main result.
\begin{prop} \label{prop:subMainThm}  
	Suppose $T$ is a tree on $X$ which displays a known tripartition $A|B|C$ 
	corresponding to vertex $r$ with $|A|,|B|\geq3$, $|C|\geq 1$. 
	If $\kappa=20$ and $m\leq 77$ then both $T$ and the numerical parameters of the PM$(T, \kappa, m)$ 
	model are generically identifiable, up to arbitrary rescaling of the tree and the exchangeability matrix $R$.
\end{prop}

\begin{proof}
	Using the notation and result of Lemma \ref{prop:ranksBigEnough}, 
	if a distribution $P$ comes from generic parameters of $PM(T,\kappa,m)$, then
	$$
	\operatorname{Flat}_{A|B|C}(P)=[\overline{M}_A,{M}_B,{M}_C],
	$$
	where $\overline{M}_A,{M}_B$ have full Kruskal rank and $M_C$ has Kruskal rank 
	at least 2. Thus equation \eqref{eq:Kruskal} of Theorem \ref{thm:Kruskal} 
	is satisfied with $l=m\kappa$, and $\overline{M}_A, {M}_B, {M}_C$ are determined
	uniquely up to simultaneous permutation and scaling of the rows. 
	
	Also, by factoring out row sums from the matrices, we can generically identify the root distribution vector $\boldsymbol \Pi$ at the 
	node $r$ and $M_A, M_B, M_C$ up to simultaneous permutation of the entries of $\boldsymbol \Pi$ and the rows of the matrices. 
	Considering any entry of $\boldsymbol \Pi$, 
	and supposing that this corresponds to an unknown class $u \in [m]$ and state $w \in [\kappa]$, then the same rows of ${M}_A,{M}_B,{M}_C$ correspond 
	to the same class $u$ and state $w$. Since Kruskal's theorem yields identifiability only up to permutation, we 
	must determine which of the $m\kappa$ rows of ${M}_A,{M}_B,{M}_C$ correspond to the same fixed class $u$. 
	
	Consider first the special case that $|A|=3$ where $A=\{a,b,c\}$. Then $T$, which is generically binary, has a subtrees rooted at $r$, with 
	leaves $A=\{x,y,z\}$ as shown in Figure \ref{fig:tree3tax}, though we do not know which two taxa from $a,b,c$ form the cherry $\{y,z\}$.
	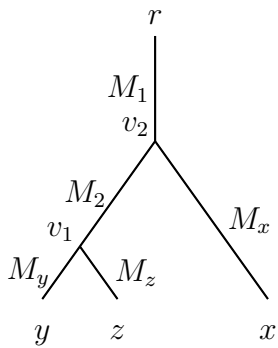
\begin{figure}[h]
		\begin{center}
			\begin{tikzpicture}
			\def\sc{1}

			\def\vdist{\sc*.7}; 
			\def\ldist{\sc*0.5};

			\foreach \i in {1,...,6} {
				\coordinate (v\i) at (\sc*\i,0);
			}
			
			\foreach[count = \i] \a in { $y$ , $z$ , $ $ ,$x$ } {
				\node[below = 0.2 cm of v\i] {\a};
			}
			
			\coordinate[above right=\vdist and \ldist of v1] (a);
			\coordinate[above right=2*\vdist and 2*\ldist of a] (b);
			\coordinate[above =2*\vdist and 2*\ldist of b] (root);

			\foreach \i/\j in {a/v1, a/v2, b/v4,  a/b, b/root} {
				\draw[thick] (\i) -- (\j);
			}

			\node[above left=-0.1 of b] {$v_{2}$};
			\node[above left=-0.1 of a] {$v_{1}$};
			\node[above=0.01 of root] {$r$};

			\path 
			(a) edge[ left=10]  node[xshift=0
			]  {$M_y$} (v1)
			(a) edge[ right=10]  node[xshift=2]  {$M_z$} (v2)
			(b) edge[ right=10]  node[xshift=2]  {$M_x$} (v4)
			(a) edge[ left=10]  node[xshift=0]  {$M_2$} (b)
			(b) edge[ left=10]  node[xshift=2]  {$M_1$} (root)
			;
			\end{tikzpicture}
		\end{center}	
		\caption{A subtree of $T$ with leaves $A=\{a,b,c\}=\{x,y,z\}$.}\label{fig:tree3tax}
	\end{figure}
	
	The Markov matrix $M_A$ is of size $m\kappa\times\kappa^3$. Choose the $\ell^{th}$ row of 
	$M_A$ where $\ell=(u,w)$ for unknown $u, w$. It is a row vector with $\kappa^3$ entries, but 
	we can reconfigure it as a 3-dimensional tensor of size
	$\kappa\times\kappa\times\kappa$ so its $(i, j, k)$-entry
	is $P(a=i, \, b=j, \, c=k \mid r=\ell)$. 
	Since the PM model is time reversible, take $v_1$ as the root of the subtree in 
	Figure \ref{fig:tree3tax}. Then for unknown $1\times\kappa$ vector ${\boldsymbol \pi}_{v_1}$, and $\kappa\times\kappa$ 
	Markov matrices $M_x, M_y, M_z, M_1, M_2$ for class $u$ on this subtree, the joint distribution of states at 
	$x, y, z, r$  for fixed class $u$ is
	\begin{align*}
	P(x=i, \, y=j, \, &z=k, \, r=(u, w))\\
	&=\sum_{\alpha=1}^\kappa \sum_{\beta=1}^\kappa 
	{\boldsymbol \pi}_{v_1}(\beta) M_y(\beta,j)M_z(\beta, k)M_2(\beta,\alpha)M_1(\alpha,w)M_x(\alpha,i)\\
    	&=\sum_{\beta=1}^\kappa {\boldsymbol \pi}_{v_1}(\beta) M_y(\beta,j)M_z(\beta, k) 
	\left(\sum_{\alpha=1}^\kappa M_2(\beta,\alpha)M_1(\alpha,w)M_x(\alpha,i) \right)\\
    	&=\sum_{\beta=1}^\kappa {\boldsymbol \pi}_{v_1}(\beta) M_y(\beta,j)M_z(\beta, k) \widehat{M}_{(u,w)}(\beta,i)=[ {\boldsymbol \pi}_{v_1}; M_y, M_z, \widehat{M}_{(u,w)}],
	\end{align*}
	where $\widehat{M}_{(u,w)}=M_2 \diag(M_1(\cdot,w)) M_x$ with $M_1(\cdot,w)$ denoting the
	$w^{th}$ column of $M_1$. 
	For fixed $u$ this is simply a 
	rescaling of the conditional distribution $P(x=i, \, y=j, \, z=k \mid r=(u,w))$ given in the $\ell^{th}$ row of $M_A$.
	
	Thus applying Kruskal's theorem to each row of $M_A$ reshaped into such a $3$-way tensor, we can decompose 
	$P(x=i, y=j, z=k \mid  r=\ell)$ for each $\ell=(u,w)$ into a triple product, as the matrices 
	generically all have rank $\kappa$. Note that for each $\ell
	= (u, w)$, Kruskal's theorem gives the matrices $M_y, M_z,\widehat{M}_{(u,w)}$ up to 
	ordering of their $\kappa$ rows. Two of these matrices, $M_y, M_z$, will be dependent only on the class $u$, 
	but not the state $w$. So considering all $\ell=(u,w)$, we can find $\kappa$ rows of $M_A$ with 
	the same (possibly permuted rows) version of $M_y$ and $M_z$ which correspond to a single class $u$. In this 
	way we can group the rows of $M_A, M_B, M_C$ with entries of $\boldsymbol \Pi$ by class $u$. Now taking those 
	rows of ${M}_A,{M}_B,{M}_C$, and entries of $\boldsymbol \Pi$ for one class $u$ and reassembling them in a 
	$3$-way product gives a tensor for a single class GTR model on the full tree $T$. Both the tree $T$
	and numerical parameters are identifiable for this single-class model by Theorem \ref{thm:GTRid}.
	
	For the general case, suppose $|A|, |B|\geq 3$. Then by marginalization down to $|A|=3$
	we can identify the subtrees and parameters for $B, C$. Then interchanging the roles of $A$ and $B$ 
	identifies the subtree and parameters for $A$.
\end{proof}

Combining Proposition \ref{prop:IDtripartion} with Proposition \ref{prop:subMainThm}, we have 
proved the main result.

\begin{theorem}\label{thm:mainThm}
	Let $T$ be a tree with at least $9$ taxa. Then under the PM $(T, 20, m)$ model with $m<74$, 
	both $T$ and numerical parameters are generically identifiable, up to arbitrary rescaling of the tree and the exchangeability matrix $R$.
\end{theorem}

Theorem \ref{thm:mainThm} extends to certain tree shapes with fewer than 9 taxa.  To apply 
Proposition \ref{prop:subMainThm}, $T$ must display a tripartition
with two of its subsets of size at least 3, so that $T$ must have at least 7 taxa.  Such a 
tripartition will be generically  identifiable by the argument given for Proposition \ref{prop:IDtripartion}.
\begin{cor} For the profile mixture model $PM(T,20,m)$ with $m < 74$, parameters are generically identifiable
	if $T$ has any of the $8$-taxon tree shapes (a)-(d) shown in Figure \ref{fig:8taxTrees}, or the $7$-taxon caterpillar shape.
\end{cor}

\begin{figure}[h]
	$(a)$ \includegraphics[width=2in]{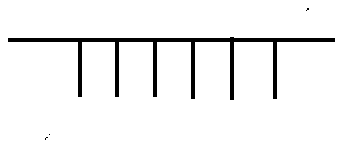}
	\hspace*{0.2in}
	$(b)$ \includegraphics[width=1.5in]{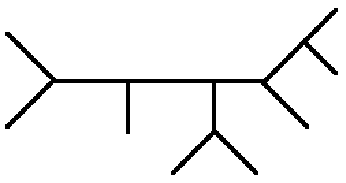}
	\hspace*{0.2in}
	$(c)$ \includegraphics[width=1.5in]{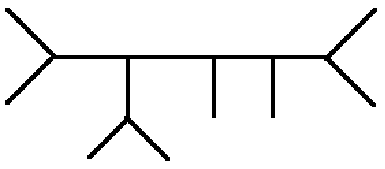}
	\hspace*{0.2in}
	$(d)$ \includegraphics[width=1.5in]{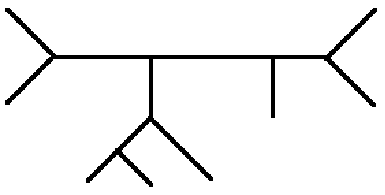}
	\hspace*{0.2in}
	$(e)$ \includegraphics[width=1.5in]{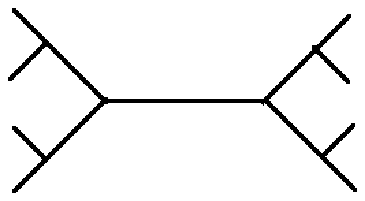}
	\caption{All binary unrooted tree shapes for 8 taxa. Parameters of the PM model are generically 
		identifiable for trees (a)-(d). The arguments of this paper do not answer the identifiability question for tree (e).}
\label{fig:8taxTrees}
\end{figure}

\section*{Acknowledgments}  This research was supported, in part, by the National Institutes of Health Grant 
R01 GM117590, awarded under the Joint DMS/NIGMS Initiative to Support Research at the Interface of the 
Biological and Mathematical Sciences.

\section*{Author Disclosure Statement}  No competing financial interests exist.

\ifthenelse{\boolean{submittedVersion}}{

}{
\bibliographystyle{plainnat}      
\bibliography{identifiability}
}

\end{document}